\documentclass[10pt]{article}

\usepackage{authblk}
\usepackage{amsmath}
\usepackage{amsthm}
\usepackage{amssymb}
\usepackage{amsfonts}
\usepackage[breaklinks,colorlinks]{hyperref}
\usepackage{enumerate}
\usepackage{enumitem}
\usepackage{float}
\usepackage{graphicx}
\usepackage{latexsym}
\usepackage{mathtools}
\usepackage{mathptm}
\usepackage{url}
\usepackage{fullpage}
\usepackage{ifthen}
\usepackage{color}
\usepackage{verbatim}
\allowdisplaybreaks

\newtheorem{theorem}{Theorem}
\newtheorem{definition}[theorem]{Definition}
\newtheorem{lemma}[theorem]{Lemma}

\newcommand{\eq}[1]{\hyperref[eq:#1]{(\ref*{eq:#1})}}
\renewcommand{\sec}[1]{\texorpdfstring{\hyperref[sec:#1]{Section~\ref*{sec:#1}}}{Section~\ref*{sec:#1}}}
\newcommand{\thm}[1]{\texorpdfstring{\hyperref[thm:#1]{Theorem~\ref*{thm:#1}}}{Theorem~\ref*{thm:#1}}}
\newcommand{\lem}[1]{\hyperref[lem:#1]{Lemma~\ref*{lem:#1}}}
\newcommand{\cor}[1]{\hyperref[cor:#1]{Corollary~\ref*{cor:#1}}}
\newcommand{\fig}[1]{\hyperref[fig:#1]{Figure~\ref*{fig:#1}}}
\newcommand{\apd}[1]{\hyperref[apd:#1]{Appendix~\ref*{apd:#1}}}

\floatstyle{ruled}
\newfloat{algorithm}{htb}{loa}
\floatname{algorithm}{Algorithm}

\newcommand{\beq}{\begin{equation}}
\newcommand{\eeq}{\end{equation}}
\newcommand{\ba}{\begin{array}}
\newcommand{\ea}{\end{array}}
\newcommand{\bal}{\begin{align}}
\newcommand{\eal}{\end{align}}
\newcommand{\bc}{\begin{cases}}
\newcommand{\ec}{\end{cases}}
\newcommand{\bpm}{\begin{pmatrix}}
\newcommand{\epm}{\end{pmatrix}}
\newcommand{\ben}{\begin{enumerate}}
\newcommand{\een}{\end{enumerate}}
\newcommand{\bit}{\begin{itemize}}
\newcommand{\eit}{\end{itemize}}

\newcommand{\tsf}{\textsf}

\newcommand{\bbC}{\mathbb{C}}

\newcommand{\bbR}{\mathbb{R}}

\newcommand{\bbZ}{\mathbb{Z}}

\renewcommand{\Re}{\mathrm{Re}}

\let\originalleft\left
\let\originalright\right
\renewcommand{\left}{\mathopen{}\mathclose\bgroup\originalleft}
\renewcommand{\right}{\aftergroup\egroup\originalright}

\newcommand{\lb}{\left(}
\newcommand{\rb}{\right)}

\newcommand{\defeq}{\vcentcolon=}

\newcommand{\abs}[1]{\left | #1 \right |}

\newcommand{\dsum}{\displaystyle\sum\limits}

\newcommand{\diag}[1]{\operatorname{diag}\lb #1 \rb}

\newcommand{\nm}[1]{\left \| #1  \right \|}

\newcommand{\bnm}[1]{\| #1 \|}

\newcommand{\poly}[1]{\operatorname{poly}\lb #1 \rb}
\newcommand{\bgo}[1]{O \lb #1 \rb}
\newcommand{\omg}[1]{\Omega\lb #1 \rb}

\renewcommand{\log}[1]{\operatorname{log} \lb #1 \rb}
\newcommand{\loglog}[1]{\operatorname{log}\operatorname{log} \lb #1 \rb}

\renewcommand{\exp}[1]{\operatorname{exp} \lb #1 \rb}

\newcommand{\ket}[1]{\left | #1\right\rangle}
\newcommand{\bra}[1]{\left\langle #1\right|}
\newcommand{\braket}[2]{\left\langle #1|#2 \right\rangle}
\newcommand{\ketbra}[2]{\left|#1\rangle \langle #2  \right |}

\newcommand{\PSPACE}{\tsf{PSPACE}}

\newcommand{\PP}{\tsf{PP}}

\newcommand{\BQP}{\tsf{BQP}}

\renewcommand{\phi}{\varphi}

\newcommand{\inn}{\mathrm{in}}

\begin{document}

\title{Quantum algorithm for linear differential equations with exponentially improved dependence on precision}

\author{Dominic W.\ Berry \thanks{Department of Physics and Astronomy, Macquarie University}
\hspace{1em}
Andrew M.\ Childs \thanks{Department of Computer Science and Institute for Advanced Computer Studies, University of Maryland}$^{\,\,\,,}$\thanks{Joint Center for Quantum Information and Computer Science, University of Maryland}
\hspace{1em}
Aaron Ostrander $^{\ddag,}$\thanks{Department of Physics, University of Maryland}
\hspace{1em}
Guoming Wang $^{\ddag}$
}

\date{}

\maketitle

%%%%%%%%%%%%%%%%%%%%%%%%%%%%%%%%%%%%%%%%%%%%%%%%%%%%%%%%%%%%%%%%%%%%%%%%%%%%%%

\begin{abstract}
We present a quantum algorithm for systems of (possibly inhomogeneous) linear ordinary differential equations with constant coefficients. The algorithm produces a quantum state that is proportional to the solution at a desired final time. The complexity of the algorithm is polynomial in the logarithm of the inverse error, an exponential improvement over previous quantum algorithms for this problem. Our result builds upon recent advances in quantum linear systems algorithms by encoding the simulation into a sparse, well-conditioned linear system that approximates evolution according to the propagator using a Taylor series. Unlike with finite difference methods, our approach does not require additional hypotheses to ensure numerical stability. 
\end{abstract}

%%%%%%%%%%%%%%%%%%%%%%%%%%%%%%%%%%%%%%%%%%%%%%%%%%%%%%%%%%%%%%%%%%%%%%%%%%%%%%
%%%%%%%%%%%%%%%%%%%%%%%%%%%%%%%%%%%%%%%%%%%%%%%%%%%%%%%%%%%%%%%%%%%%%%%%%%%%%%
\section{Introduction}
\label{intro}

One of the original motivations for developing a quantum computer was to efficiently simulate Hamiltonian dynamics, i.e., differential equations of the form $\frac{d \vec x}{dt} = A \vec x$ where $A$ is anti-Hermitian. Given a suitable description of $A$, a copy of the initial quantum state $\ket{x(0)}$, and an evolution time $T$, the goal is to produce a quantum state that is $\epsilon$-close to the final state $\ket{x(T)}$. The first algorithms for this problem had complexity polynomial in $1/\epsilon$ \cite{LLo96,AT03,Chi04,BACS07}. Subsequent work gave an algorithm with complexity $\poly{\log{1/\epsilon}}$---an exponential improvement---which is optimal in a black-box model \cite{berry2014exponential}. More recent work has streamlined these algorithms and improved their dependence on other parameters \cite{berry2015simulating,berry2015hamiltonian,low2016optimal,berry2016corrected,low2016hamiltonian,NB16}.

While Hamiltonian simulation has been a focus of quantum algorithms research, the more general problem of simulating linear differential equations of the form $ \frac{d \vec x}{dt} = A \vec x + \vec b$ for arbitrary (sparse) $A$ is less well studied. Reference~\cite{berry2014high} solves this problem using a quantum linear systems algorithm (QLSA) to implement linear multistep methods, which represent the differential equations with a system of linear equations by discretizing time.  The complexity of this approach is $\poly{1 / \epsilon}$.  Considering the recent improvements to the complexity of Hamiltonian simulation, it is natural to ask whether linear differential equations can be solved more efficiently as a function of $\epsilon$.

Hamiltonian simulation is a central component of the QLSA, and the techniques underlying $\poly{\log{1/\epsilon}}$ Hamiltonian simulation have been adapted to give a QLSA with complexity $\poly{\log{1/\epsilon}}$ \cite{childs2015quantum}.  However, even if this improved QLSA is used to implement the algorithm of Ref.~\cite{berry2014high}, the overall complexity is still $\poly{1 / \epsilon}$, since the multistep method itself is a significant source of error.

In a similar vein, the QLSA of Ref.~\cite{harrow2009quantum} has $\poly{1 / \epsilon}$ complexity even when using a Hamiltonian simulation algorithm with $\poly{\log{1/ \epsilon }}$ complexity, simply because phase estimation has complexity $\poly{1 / \epsilon}$. 
Reference~\cite{childs2015quantum} provides a QLSA with $\poly{\log{1/ \epsilon }}$ complexity by avoiding phase estimation and instead directly inverting the linear system using a linear combination of unitaries (LCU). Thus, one might consider realizing the solution of $\frac{d \vec x}{dt}=A \vec x + \vec b$ as a linear combination of unitaries. Unfortunately, in the general case where $A$ is not anti-Hermitian, the best implementation of this approach that we are aware of has an exponentially small success probability.

In this paper, we circumvent these limitations and present a quantum algorithm for linear differential equations with complexity $\poly{\log{1/\epsilon}}$, an exponential improvement over Ref.~\cite{berry2014high}. As in Ref.~\cite{berry2014high}, our approach applies the QLSA.  However, instead of using a linear multistep method, we encode a truncation of the Taylor series of $\exp{At}$, the propagator for the differential equation, into a linear system.  Since it effectively implements a linear combination of operations, our approach is conceptually similar to quantum simulation via linear combinations of unitaries, but we achieve significantly better performance by constructing this linear combination stepwise through a system of linear equations.  This alternative to direct application of LCU methods might be advantageous for other quantum algorithms.

In addition to scaling well with the simulation error, our algorithm has favorable performance as a function of other parameters.  The complexity is nearly linear in the evolution time, which is a quadratic improvement over Ref.~\cite{berry2014high} and is nearly optimal \cite{BACS07}.  The complexity is also nearly linear in the sparsity of $A$ and in a parameter characterizing the decay of the solution vector.  The latter dependence is necessary since producing a normalized version of a subnormalized solution vector is equivalent to postselection, which is computationally intractable \cite{aaronson2005quantum}, as discussed further in \sec{discussion}.  Along similar lines, we assume that the eigenvalues of $A$ have non-positive real part since it is intractable to simulate exponentially growing solutions.  (This improves upon Ref.~\cite{berry2014high}, where the eigenvalues $\lambda$ of $A$ must satisfy $|\arg(-\lambda)|\le\alpha$ for some constant $\alpha$ depending on the stability of the multistep method.)  For a precise statement of the main result, see \thm{main}.

This paper is organized as follows. \sec{ls} describes how we encode the solution of a system of differential equations into a system of linear equations.  The following three sections analyze properties of this system: \sec{cond} bounds its condition number, \sec{solerr} analyzes how well it approximates the differential equation, and \sec{succprob} shows that a measurement of its solution vector provides a solution of the differential equation with appreciable probability. In \sec{stateprep} we explain how to prepare the state that is input to the QLSA using black boxes for the initial condition and inhomogeneous term of the differential equation. We formally state and prove our main result in \sec{main}. Finally, we conclude in \sec{discussion} with a discussion of the result and some open problems.

%%%%%%%%%%%%%%%%%%%%%%%%%%%%%%%%%%%%%%%%%%%%%%%%%%%%%%%%%%%%%%%%%%%%%%%%%%%%%%
%%%%%%%%%%%%%%%%%%%%%%%%%%%%%%%%%%%%%%%%%%%%%%%%%%%%%%%%%%%%%%%%%%%%%%%%%%%%%%
\section{Constructing the Linear System}
\label{sec:ls}

As in Ref.~\cite{berry2014high} we consider a differential equation of the form
\begin{equation}
\displaystyle\frac{d \vec x}{dt} = A \vec x + \vec b
\label{eq:dfeq}
\end{equation}
where $A$ and $\vec b$ are time-independent. This has the exact solution  
\beq
\vec x(t) = \exp{At} \vec x(0) +  (\exp{At} - I) A^{-1}\vec b.
\label{eq:dfeqsol}
\eeq

Define
\beq
  T_k(z)
  \defeq \dsum_{j=0}^k \dfrac{z^j}{j!}
  \approx \exp{z}
\label{eq:Tk}
\eeq
and
\beq
  S_{k}(z)
  \defeq \displaystyle\sum_{j=1}^{k} \frac{z^{j-1}}{j!}
  \approx (\exp{z}-1)z^{-1}
\label{eq:Sk}
\eeq
where the approximations hold for large $k$. Then for short evolution time $h$ (namely $h \le 1/\nm{A}$, where $\nm{\cdot}$ denotes the spectral norm) and large $k$, we can approximate the solution by 
\beq
\vec x (h) \approx T_k (Ah) \vec x(0) +  S_k(Ah) h \vec b.
\eeq
This approximate solution can be used in turn as an initial condition for another step of evolution, and we can repeat this procedure as desired for a total number of steps $m$.

We encode this procedure in a linear system using the following family of matrices.

\begin{definition}
Let $A$ be an $N \times N$ matrix, and let $m, k, p \in \bbZ^+$. Define 
\begin{align}
C_{m,k,p}(A) &\defeq \dsum_{j=0}^{d} \ketbra{j}{j} \otimes I  -\dsum_{i=0}^{m-1} \dsum_{j=1}^{k} \ketbra{i(k+1)+j}{i(k+1)+j-1} \otimes A/j \nonumber\\*
&\quad -\dsum_{i=0}^{m-1}\dsum_{j=0}^{k} \ketbra{(i+1)(k+1)}{i(k+1)+j} \otimes I -\dsum_{j=d-p+1}^{d} \ketbra{j}{j-1} \otimes I,
\end{align}
where $d \defeq m(k+1)+p$, and $I$ is the $N \times N$ identity matrix.
\end{definition}

Now consider the linear system
\beq
C_{m, k, p} (Ah) \ket{x} = \ket{0} \ket{x_{\inn}} + h \displaystyle\sum_{i=0}^{m-1} | i(k+1) +1 \rangle  \ket{ b},
\label{eq:ls}
\eeq
where $\ket{ x_{\inn}}, \ket{b} \in \bbC^N$ and $h \in \bbR^+$. The first register labels a natural block structure for $C_{m,k,p}$. For example, the system $C_{2,3,2}(Ah) | x \rangle = \ket{0} \ket{x_{\inn}} + h \sum_{i=0}^{1} | 4i +1 \rangle  \ket{ b} $ is as follows:
\beq
C_{2,3,2}(Ah) | x \rangle = \bpm
I & & & & & & & & & &\\
-Ah & I & & & & & & & & &\\
& -{Ah}/{2} & I & & & & & & & &\\
& & -{Ah}/{3} & I & & & & & & &\\
-I & -I & -I & -I  & I & & & & & &\\
& & & & -Ah & I & & & & &\\
& & & & & -Ah/2 & I & & & &\\
& & & & & & -Ah/3 & I &  & &\\
& & & & -I & -I & -I & -I & I & & \\
& & & &  & & & &  -I & I &\\
& & & &  & & & &  & -I & I
\epm
| x \rangle = 
\bpm
| x_{\inn} \rangle \\
h | b \rangle \\
0 \\
0\\
0\\
h | b \rangle \\
0 \\
0 \\
0 \\
0 \\
0 \\
\epm
.
\eeq
After performing $m$ steps of the evolution approximated with a Taylor series of order $k$, the solution is kept constant for $p$ steps.
This ensures a significant probability of obtaining the solution at the final time, similarly as in Ref.~\cite{berry2014high}.
Note that $C_{m,k,p}(Ah)$ is nonsingular, since it is a lower-triangular matrix with nonzero diagonal entries. If $N=1$, then $C_{m,k,p}(Ah)$ is a $(d+1) \times (d+1)$ matrix.

The solution of Eq.~(\ref{eq:ls}) is
\beq
\ket{x}=C_{m,k,p}(Ah)^{-1} \left[ \ket{0}\ket{x_{\inn}} + h \displaystyle\sum_{i=0}^{m-1} | i(k+1) +1 \rangle  \ket{ b} \right], 
\label{eq:lssol1}
\eeq
which can be written as
\beq
\ket{x} =\dsum_{i=0}^{m-1} \sum\limits_{j=0}^{k} \ket{i(k+1)+j} \ket{x_{i,j}} + \dsum_{j=0}^{p} \ket{m(k+1)+j} \ket{x_{m,j}} \label{eq:lssol2}
\eeq
for some $\ket{x_{i,j}} \in \bbC^N$. By the definition of $C_{m,k,p}(Ah)$, these $\ket{x_{i,j}}$s satisfy
\begin{alignat}{2}
\ket{x_{0, 0}} &=  \ket{x_{\inn}}, \\
\ket{x_{i, 0}} &= \dsum_{j=0}^{k} \ket{x_{i-1, j}}, &\qquad&  1 \le i \le m, \\
\ket{x_{i, 1}} &= Ah \ket{x_{i, 0}}+h\ket{b},  &&  0 \le i < m, \\
\ket{x_{i, j}} &= (Ah/j) \ket{x_{i, j-1}},  &&  0 \le i < m, ~2 \le j \le k, \\
\ket{x_{m, j}} &= \ket{x_{m, j-1}},  && 1 \le j \le p. 
\end{alignat}
From these equations, we obtain
\begin{alignat}{2}
\ket{x_{0,0}} &= \ket{x_{\inn}}, \\
\ket{x_{0,j}} &= ((Ah)^j/j!) \ket{x_{0,0}} + ((Ah)^{j-1}/j!) h \ket{b},  &\qquad& 1 \le j \le k, \\
\ket{x_{1,0}} &= T_k(Ah) \ket{x_{0,0}} + S_k(Ah) h \ket{b}\nonumber \\
&\approx \exp{Ah} \ket{x_{\inn}}+ (\exp{Ah} - I )A^{-1}  \ket{b} , \\
\ket{x_{1,j}} &= ((Ah)^j/j!) \ket{x_{1,0}} + ((Ah)^{j-1}/j!) h \ket{b}, && 1 \le j \le k, \\ 
\ket{x_{2,0}} &= T_k(Ah) \ket{x_{1,0}} + S_k(Ah) h \ket{b}\nonumber \\
& \approx \exp{2Ah} \ket{x_{\inn}} + (\exp{2Ah}-I)A^{-1} \ket{b}, \\
&~~\vdots \nonumber \displaybreak[0] \\
\ket{x_{m-1,0}} &= T_k(Ah)\ket{x_{m-2,0}} + S_k(Ah) h \ket{b}\nonumber \\
&\approx \exp{Ah(m-1)} \ket{x_{\inn}} + (\exp{Ah(m-1)}-I) A^{-1} \ket{b}, \\
\ket{x_{m-1,j}} &= ((Ah)^j/j!) \ket{x_{m-1,0}} + ((Ah)^{j-1}/j!) h \ket{b}, && 1 \le j \le k, \\ 
\ket{x_{m,0}} &= T_k(Ah)\ket{x_{m-1,0}} + S_k(Ah) h \ket{b} \nonumber \\
&\approx \exp{Ahm} \ket{x_{\inn}} + (\exp{Ahm}-I) A^{-1} \ket{b}, \\
\ket{x_{m,j}} &= \ket{x_{m,0}}\nonumber \\
&\approx \exp{Ahm} \ket{x_{\inn}} + (\exp{Ahm}-I) A^{-1} \ket{b}, && 1 \leq j \leq p.
\label{historystate1}
\end{alignat}
In these approximations, we assume $k$ is sufficiently large that we can neglect the truncation errors $\nm{T_k(Ah)-\exp{Ah}}$ and $\nm{S_k(Ah)h -(\exp{Ah}-I) A^{-1} }$ (we make this more precise in \sec{solerr}). Note that $\ket{x}$ (defined by Eq.~(\ref{eq:lssol2})) includes a piece that can be interpreted as the \emph{history state} of the evolution 
\beq
\dfrac{d \vec x}{dt}=A \vec x+ \vec b
\label{eq:dfeq2}
\eeq
with the initial condition $\vec x(0)= \vec x_{\inn}$. More precisely, $\ket{x_{i,0}}$ is a good approximation of the system's state at time $ih$, for any $i \in \{0,1,\dots, m\}$.  Furthermore, $\ket{x_{m,0}}=\ket{x_{m,1}}=\dots=\ket{x_{m,p}}$ is a good approximation of 
\beq
\vec x(t) = \exp{At}  \vec x_{\inn} + (\exp{At}-I)  A^{-1}\vec b
\label{eq:dfeqsol2}
\eeq
for $t=mh$. 

If we measured the first register of the state $\ket{x}/\nm{\ket{x}}$ in the standard basis, we would obtain the state $\ket{x_{i,j}}/\nm{\ket{x_{i,j}}}$ for random $i,j$. By choosing a large $p$, we can ensure there is a high probability of obtaining $\ket{x_{m,0}}/\nm{\ket{x_{m,0}}}$, $\ket{x_{m,1}}/\nm{\ket{x_{m,1}}}$, $\dots$, or $\ket{x_{m,p}}/\nm{\ket{x_{m,p}}}$ (as we show in \sec{succprob}). Then this probability can be raised to $\omg{1}$ by using amplitude amplification (or by classical repetition). This is how we prepare a state close to $\vec x(t)/\nm{\vec x(t)}$ for $t=mh$. 

Note that the state we generate is not of the same form as the history state in Ref.~\cite{berry2014high}, which only encodes $\vec x(t)$ at intermediate times. The solution of our linear system not only encodes $\vec x (t)$ at intermediate times (via $| x_{i,0} \rangle$) but also encodes $((Ah)^j/j!) \vec x(t) + ((Ah)^{j-1}/j!) h \vec b$ at intermediate times (via $|x_{i,j} \rangle $).

To analyze the performance of this approach to solving differential equations, we establish three properties of this system of linear equations.  First, since the complexity of the best known QLSAs grows linearly with condition number (and sublinear complexity is impossible unless $\BQP=\PSPACE$ \cite{harrow2009quantum}), we analyze the condition number of $C_{m,k,p}$ (\sec{cond}).  Second, we show that the solution of the linear system includes a piece that is close to the solution of the associated differential equation (\sec{solerr}).  Third, we show that this piece can be obtained from a measurement that succeeds with appreciable probability (\sec{succprob}).

%%%%%%%%%%%%%%%%%%%%%%%%%%%%%%%%%%%%%%%%%%%%%%%%%%%%%%%%%%%%%%%%%%%%%%%%%%%%%%
%%%%%%%%%%%%%%%%%%%%%%%%%%%%%%%%%%%%%%%%%%%%%%%%%%%%%%%%%%%%%%%%%%%%%%%%%%%%%%
\section{Condition Number}
\label{sec:cond}

In this section, we upper bound the condition number of the matrix $C_{m,k,p}(A)$ under mild assumptions about $A$. We begin with a technical lemma that upper bounds the norms of the columns of the inverse of this matrix.

\begin{lemma}
Let $\lambda \in \bbC$ such that $\abs{\lambda}\le 1$ and $\Re(\lambda)\le 0$. Let $m, k, p \in \bbZ^+$ such that $k\ge 5$ and $(k+1)!\ge 2m$, and let $d=m(k+1)+p$. Then for any $n, l \in \{0,1,\dots, d\}$, 
\beq
\nm{C_{m,k,p}(\lambda)^{-1} \ket{l}} \le \sqrt{1.04 e I_0(2)(m+p)}
\eeq
with $I_0(2)<2.28$ a modified Bessel function of the first kind, and
\beq
\abs{\bra{n} C_{m,k,p}(\lambda)^{-1} \ket{l}} \le \sqrt{1.04 e}.
\eeq
\label{lem:cond}
\end{lemma}

\begin{proof}
Recall the definitions of $T_k(z)$ and $S_{k}(z)$ in Eqs.~(\ref{eq:Tk}) and (\ref{eq:Sk}), respectively.  We also define
\beq
  T_{b,k}(z) \defeq \sum_{j=b}^k \frac{b! z^{j-b}}{j!}
\eeq
for $b\le k$.

Fix any $l \in \{0,1,\dots,d\}$. Suppose the solution of the linear system
\beq
C_{m,k,p}(\lambda) \ket{x} = \ket{l}
\eeq 
is
\beq
\ket{x}=\dsum_{i=0}^{m-1} \sum\limits_{j=0}^{k} x_{i,j} \ket{i(k+1)+j}  + \dsum_{j=0}^{p} x_{m,j} \ket{m(k+1)+j} 
\eeq
for some $x_{i,j} \in \bbC$. By the definition of $C_{m,k,p}(\lambda)$, the $x_{i,j}$s should satisfy 
\begin{alignat}{2}
x_{i, 0} - \dsum_{j=0}^{k} x_{i-1, j} &= \delta_{i(k+1), l}, &\qquad& 1 \le i \le m, \\
x_{i, j} - (\lambda/j) x_{i, j-1} &= \delta_{i(k+1)+j, l}, && 0 \le i < m,~ 1 \le j \le k, \\
x_{m, j} - x_{m, j-1} &= \delta_{m(k+1)+j, l}, && 1 \le j \le p, 
\label{eq:xij}
\end{alignat}
where $\delta_{i,j}=1$ if $i=j$, and $0$ otherwise. 

We consider the cases $0\le l < m(k+1)$ and $m(k+1) \le l \le d$ separately.
\bit
\item{Case 1: $0 \le l < m(k+1)$.} Suppose $l=a(k+1)+b$ for some $0 \le a < m$ and $0 \le b \le k$.  In this case, Eq.~(\ref{eq:xij}) implies 
\begin{alignat}{2}
{x_{i,j}} &= 0,  &\qquad& 0 \le i < a, ~ 0 \le j \le k, \\ 
{x_{a,j}} &= 0,  &&  0 \le j < b, \\  
{x_{a,j}} &= b! \lambda^{j-b}/j! , &&  b \le j \le k, \\ 
{x_{a+1,0}} &= T_{b,k}(\lambda), \\
{x_{a+1,j}} &= (\lambda^j / j!) {x_{a+1,0}}, && 1 \le j \le k,\\
{x_{a+2,0}} &=T_k(\lambda) {x_{a+1, 0}}=T_{k}(\lambda) T_{b,k}(\lambda), \\
&~~\vdots \nonumber \\
{x_{m,0}} &= T_k(\lambda) {x_{m-1,0}}
=(T_k(\lambda))^{m-a-1} T_{b,k}(\lambda), \\
{x_{m,j}} &= {x_{m, 0}}=(T_k(\lambda))^{m-a-1} T_{b,k}(\lambda), && 1 \le j \le p.
\end{alignat}
Since $\abs{\lambda} \le 1$, for any $b \le j \le k$, we have
\beq
\abs{x_{a,j}} = b! \abs{\lambda}^{j-b}/j! \le b!/j! \le 1. 
\eeq
Furthermore, since $\abs{\lambda}\le 1$ and $\Re(\lambda)\le 0$, by \lem{taylor1} and \lem{taylor2} in \apd{taylor}, we have
\beq
\abs{T_{k}(\lambda)} \le 1+\frac 1{(k+1)!} \le 1+\frac 1{2m},
\eeq
\beq
\abs{T_{b,k}(\lambda)} \le \sqrt{1.04}.
\eeq
Consequently, we have
\beq
\abs{x_{i, 0}}=\abs{T_k(\lambda)^{i-a-1}T_{b,k}(\lambda)} 
\le  (1+1/2m)^m \sqrt{1.04}  \le  \sqrt{1.04 e}, \qquad a+1 \le i \le m,
\eeq
and
\beq
\abs{x_{i,j}} = \abs{(\lambda^j/j!) x_{i,0}} \le \abs{x_{i,0}} \le \sqrt{1.04 e},
\qquad a+1 \le i \le m, ~1 \le j \le k.
\eeq
It follows that
\beq
\abs{{x_{m, j}}} = \abs{{x_{m, 0}}}  \le \sqrt{1.04 e}, \qquad 0 \le j \le p.
\eeq

Using these facts, we obtain
\begin{align}
\nm{\ket{x}}^2
& = \dsum_{i=0}^{m-1} \dsum_{j=0}^k \abs{x_{i,j}}^2
+ \dsum_{j=0}^p \abs{{x_{m,j}}}^2\nonumber \\
& = 
\dsum_{j=b}^k \abs{b!\lambda ^{j-b}/j!}^2
+\dsum_{i=a+1}^{m-1} \dsum_{j=0}^k \abs{(\lambda^j/j!){x_{i,0}}}^2
+ (p+1) \abs{{x_{m,0}}}^2 \nonumber \\
& \le 
\dsum_{j=b}^k (b!/j!)^2 
+\dsum_{i=a+1}^{m-1} \dsum_{j=0}^k (1/j!)^2 \abs{x_{i,0}}^2
+ (p+1) \abs{{x_{m,0}}}^2 \nonumber \\
& \le 
\dsum_{j=b}^k (b!/j!)^2 
+1.04 e \dsum_{i=a+1}^{m-1} \dsum_{j=0}^k (1/j!)^2 
+ 1.04 e (p+1)\nonumber   \\
& \le  
I_0(2)+1.04 e I_0(2)(m-a-1)+1.04e(p+1)\nonumber  \\
& \le 
1.04e I_0(2)(m+p),
\label{eq:case1}
\end{align}
where in the fifth step we use the facts
\begin{alignat}{2}
\dsum_{j=0}^k (1/j!)^2  &\le \dsum_{j=0}^{\infty} (1/j!)^2 = I_0(2), \label{eq:factsquaresum} \\
\dsum_{j=b}^k (b!/j!)^2 &\le \dsum_{s=0}^{k-b} 1/(b+1)^{2s} 
  \le \frac{1}{1-(b+1)^{-2}}
  = 1 + \frac{1}{b(b+2)}
  \le \frac{4}{3} < I_0(2), &\qquad& 1 \le b \le k.
\label{eq:factsquaresum2}
\end{alignat}

\item{Case 2: $m(k+1) \le l \le d$.} Suppose $l = m(k+1)+b$ for some $0 \le b \le p$. In this case, Eq.~(\ref{eq:xij}) implies
\begin{alignat}{2}
{x_{i,j}}&=0, &\qquad&  0 \le i <m, ~  0 \le j \le k,\\
{x_{m,j}}&=0, &&  0 \le j <b,\\
{x_{m,j}}&=1, && b \le j \le p.
\end{alignat}
It follows that
\begin{align}
\nm{\ket{x}}^2 
& = \dsum_{i=0}^{m-1} \dsum_{j=0}^k \abs{{x_{i,j}}}^2
+ \dsum_{j=0}^p \abs{{x_{m,j}}}^2 \nonumber \\
& =  p-b+1 \nonumber\\
& \le  p+1.
\label{eq:case2}
\end{align}
\eit

In both of the above cases, we have $\nm{\ket{x}}\le \sqrt{1.04 e I_0(2)(m+p)}$ and $\abs{\braket{n}{x}}\le \sqrt{1.04 e}$ for any $n \in \{0,1,\dots,d\}$, as claimed.
\end{proof}

Now we are ready to upper bound the norm of the inverse of the matrix.

\begin{lemma}
Let $A=VD V^{-1}$ be a diagonalizable matrix, where $D=\diag{\lambda_0, \lambda_1,\dots,\lambda_{N-1}}$ satisfies $\abs{\lambda_i} \le 1$ and $\Re(\lambda_i) \le 0$ for $i \in \{0,1,\dots,N-1\}$. Let $m, k, p \in \bbZ^+$ such that $k\ge 5$ and $(k+1)!\ge 2m$. Then
\beq
\nm{C_{m,k,p}(A)^{-1}} \le 3 \kappa_V \sqrt{k} (m+p),
\label{eq:invnm}
\eeq
where $\kappa_V = \nm{V} \cdot \nm{V^{-1}}$ is the condition number of $V$.
\label{lem:nminv}
\end{lemma}

\begin{proof}
For convenience, we will drop the subscripts $m, k, p$, and use $C(\cdot)$ to denote $C_{m,k,p}(\cdot)$.
We diagonalize $C(A)$ as
\beq
C(A) = \tilde{V} C(D) \tilde{V}^{-1},
\eeq
where $\tilde{V} \defeq \sum_{j=0}^{d} \ketbra{j}{j} \otimes V$ has condition number $\kappa_{\tilde{V}}=\kappa_V$.
Then we may upper bound $\nm{C(A)^{-1}}$ in terms of $\nm{C(D)^{-1}}$ as
\begin{align}
\nm{C(A)^{-1}} &= \nm{\tilde{V} C(D)^{-1} \tilde{V}^{-1}}\nonumber\\
&\le \nm{\tilde{V}} \cdot \nm{C(D)^{-1}} \cdot \nm{\tilde{V}^{-1}}\nonumber \\
&= \nm{C(D)^{-1}} \cdot \kappa_{\tilde{V}}.\label{diagCA}
\end{align}

We have
\beq \label{natural}
\nm{C(D)^{-1}} = \max_{\ket{\psi}} \frac{\nm{C(D)^{-1} \ket{\psi}}}{\nm{\ket{\psi}}}
\eeq
where we maximize over all states $\ket{\psi} \in \bbC^{(d+1)N}$
The state $\ket{\psi}$ can be written as $\ket{\psi}=\sum_{l=0}^{d} \ket{l} \ket{\psi_l}$ for some $\ket{\psi_l} \in \bbC^N$.
Then we have
\begin{align}
\nm{C(D)^{-1} \ket{\psi}}^2
&=\nm{\dsum_{l=0}^d C(D)^{-1} \ket{l} \ket{\psi_l}}^2\nonumber\\
&\le (d+1) \dsum_{l=0}^d \nm{ C(D)^{-1} \ket{l} \ket{\psi_l}}^2.\label{partway}
\end{align}
Now let $\ket{\psi_l}=\sum_{j=0}^{N-1} \psi_{j,l} \ket{j}$ for some $\psi_{j,l} \in \bbC$. With $D=\sum_{j=0}^{N-1} \lambda_j \ketbra{j}{j}$, we have
$C(D)=\sum_{j=0}^{N-1} C(\lambda_j) \otimes \ketbra{j}{j}$.
Hence, by \lem{cond}, we obtain
\begin{align}
\nm{C(D)^{-1} \ket{l} \ket{\psi_l}}^2 & = 
\nm{  \dsum_{j=0}^{N-1}
 \psi_{j,l} C(\lambda_j)^{-1} \ket{l} \ket{j}}^2\nonumber \\
& = 
 \dsum_{j=0}^{N-1}
 \abs{\psi_{j,l}}^2 \nm{C(\lambda_j)^{-1} \ket{l}}^2\nonumber \\
&\le 1.04eI_0(2)(m+p) \dsum_{j=0}^{N-1}
 \abs{\psi_{j,l}}^2\nonumber  \\
& = 1.04eI_0(2)(m+p) 
\nm{\ket{\psi_l}}^2.
\label{eq:invnmstate} 
\end{align}
Using this expression in Eq.~\eqref{partway} gives  
\begin{align}
\nm{C(D)^{-1} \ket{\psi}}^2
&\le 1.04eI_0(2)(d+1)(m+p) \dsum_{l=0}^d \nm{\ket{\psi_l}}^2\nonumber \\
& = 1.04eI_0(2)(m(k+1)+p+1)(m+p) \nm{\ket{\psi}}^2\nonumber \\
& \le \frac 65 \times 1.04eI_0(2) k(m+p)^2 \nm{\ket{\psi}}^2.
\end{align}
Using this result in Eq.~\eqref{natural} yields
\beq
\nm{C(D)^{-1}} \le 3 \sqrt{k}(m+p).
\label{eq:invnmdiag}
\eeq
Combining this expression with Eq.~\eqref{diagCA} then gives Eq.~\eqref{eq:invnm}, as claimed.
\end{proof}

It remains to upper bound the norm of the matrix.

\begin{lemma}
Let $A$ be an $N \times N$ matrix such that $\nm{A} \le 1$. Let $m, k, p \in \bbZ^+$, and $k\ge 5$. Then 
\beq
\nm{C_{m,k,p}(A)} \le 2\sqrt{k}.
\eeq
\label{lem:nm}
\end{lemma}

\begin{proof}
Observe that $C\defeq C_{m,k,p}(A)$ can be written as the sum of three matrices:
\beq
C = C_1 + C_2 + C_3
\eeq
where
\beq
C_1 \defeq \dsum_{j=0}^{d} \ketbra{j}{j} \otimes I,
\eeq
\beq
C_2 \defeq -\dsum_{i=0}^{m-1} \dsum_{j=0}^{k} \ketbra{(i+1)(k+1)}{i(k+1)+j} \otimes I,
\eeq
\beq
C_3 \defeq -\dsum_{i=0}^{m-1} \dsum_{j=1}^{k} \ketbra{i(k+1)+j}{i(k+1)+j-1} \otimes A/j -\dsum_{j=d-p+1}^{d} \ketbra{j}{j-1} \otimes I,
\eeq
where $d=m(k+1)+p$. One can easily check that $\nm{C_1} = 1$, $\nm{C_2} = \sqrt{k+1}$, and $\nm{C_3} = \max\{\nm{A}, 1\}=1$ (this is trivial for $C_1$, and follows directly from a calculation of $C_2 C_2^\dag$ and $C_3 C_3^\dag$ for the other cases). Consequently,
\begin{align}
\nm{C} &\le \nm{C_1}+\nm{C_2}+\nm{C_3}\nonumber \\
& \le \sqrt{k+1} + 2\nonumber \\
& \le 2\sqrt{k}
\end{align}
as claimed.
\end{proof}

Combining \lem{nminv} and \lem{nm}, we obtain the following upper bound on the condition number of $C_{m,k,p}(A)$:
\begin{theorem}
Let $A=VD V^{-1}$ be a diagonalizable matrix such that $\nm{A} \le 1$, $D=\diag{\lambda_0, \lambda_1,\dots,\lambda_{N-1}}$ and $\Re(\lambda_i) \le 0$, for $i \in \{0,1,\dots, N-1\}$. Let $m, k, p \in \bbZ^+$ such that $k\ge 5$ and $(k+1)!\ge 2m$.  Let $C \defeq C_{m,k,p}(A)$, and let $\kappa_C=\nm{C}\cdot \nm{C^{-1}}$ be the condition number of $C$. Then 
\beq
\kappa_C \le 6 \kappa_V k(m+p),
\eeq
where $\kappa_V=\nm{V}\cdot \nm{V^{-1}}$ is the condition number of $V$.
\label{thm:cond}
\end{theorem}

%%%%%%%%%%%%%%%%%%%%%%%%%%%%%%%%%%%%%%%%%%%%%%%%%%%%%%%%%%%%%%%%%%%%%%%%%%%%%%
%%%%%%%%%%%%%%%%%%%%%%%%%%%%%%%%%%%%%%%%%%%%%%%%%%%%%%%%%%%%%%%%%%%%%%%%%%%%%%
\section{Solution Error}
\label{sec:solerr}

In this section, we prove that the solution of the linear system defined by Eq.~(\ref{eq:ls}) encodes a good approximation of the solution of the differential equation defined by Eq.~(\ref{eq:dfeq2}) with the initial condition $\vec x(0)=\vec x_{\inn}$.

\begin{theorem}
Let $A=VD V^{-1}$ be a diagonalizable matrix, where $D=\diag{\lambda_0, \lambda_1,\dots,\lambda_{N-1}}$ satisfies $\Re(\lambda_i) \le 0$ for $i \in \{0,1,\dots,N-1\}$. Let $h \in \bbR^+$ such that $\nm{Ah} \le 1$. Let $\ket{{x}_{\inn}}, \ket{b} \in \bbC^{N}$, and let $\ket{x(t)}$ be defined by Eq.~(\ref{eq:dfeqsol2}). Let $m, k, p \in \bbZ^+$ such that $k\ge 5$ and $(k+1)!\ge 2m$. Let $\ket{x_{i,j}}$ be defined by Eqs.~(\ref{eq:lssol1}) and (\ref{eq:lssol2}). Then for any $j \in \{0,1,\dots, m\}$, 
\beq
\nm{\ket{x(jh)}-\ket{x_{j,0}}}
\le 2.8 \kappa_V j  (\nm{\ket{x_{\inn}}} + mh \nm{\ket{b}})/(k+1)!,
\eeq
where $\kappa_V = \nm{V} \cdot \nm{V^{-1}}$ is the condition number of $V$.
\label{thm:solerr}
\end{theorem}
\begin{proof}
Note that $\ket{x(jh)}$ (the solution of the differential equation) satisfies the recurrence relation
\beq
\ket{ x((j+1)h) }  =  \exp{Ah} \ket{x (jh)} + ( \exp{Ah} - I ) A^{-1} \ket{ b},
\label{eq:recursion1}
\eeq
while $\ket{x_{j,0}}$ (in the solution of the associated linear system) satisfies the recurrence relation
\beq
\ket{x_{j+1,0}} = T_{k}(Ah) \ket{x_{j,0}}+S_k(Ah)h\ket{b}.
\label{eq:recursion2}
\eeq
Recall that $T_k(\lambda) = \sum_{j=0}^k \frac{\lambda^j}{j!}$ and $S_{k}(\lambda) = \sum_{j=1}^{k} \frac{\lambda^{j-1}}{j!}$. In addtion, we have $\ket{x(0)}=\ket{x_{0,0}}=\ket{x_{\inn}}$.

Define $\ket{y(t)} \defeq V^{-1} \ket{x(t)}$ and $\ket{y_{i,j}}=V^{-1}\ket{x_{i,j}}$. We will give an upper bound on $\delta_j \defeq \nm{\ket{y(jh)} - \ket{y_{j,0}}}$ and convert it into an upper bound on $\epsilon_j \defeq \nm{\ket{x(jh)} - \ket{x_{j,0}}}$.
Since $A=VDV^{-1}$, we have $\exp{Ah}=V\exp{Dh}V^{-1}$, $T_k(Ah)=VT_k(Dh)V^{-1}$, and $S_k(Ah)=VS_k(Dh)V^{-1}$. Then Eq.~(\ref{eq:recursion1}) implies
\beq
\ket{ y((j+1)h) }  =  \exp{Dh} \ket{y (jh)} + ( \exp{Dh} - I ) D^{-1} \ket{c}, 
\eeq
where $\ket{c} = V^{-1} \ket{b}$. Meanwhile, Eq.~(\ref{eq:recursion2}) implies
\beq
\ket{y_{j+1,0}} = T_{k}(Dh) \ket{y_{j,0}}+S_k(Dh)h\ket{c}.
\eeq
In addition, we have $\ket{y(0)}=\ket{y_{0,0}}=\ket{y_{\inn}} \defeq V^{-1}\ket{x_{\inn}}$. 

Now since $\Re(\lambda_i) \le 0$ for any $i \in \{0,1,\dots,N-1\}$, we have $\nm{\exp{Dh}} \le 1$. Moreover, since $\nm{Ah}\le 1$, we have $\abs{\lambda_i h} \le 1$ for any $i \in \{0,1,\dots,N-1\}$. Then since $\Re(\lambda_i h)\le 0$ for any $i \in \{0,1,\dots,N-1\}$, by \lem{taylor1} in \apd{taylor}, we get
\beq
\nm{\exp{Dh} - T_k(Dh)} \le  1/(k+1)! \, ,
\eeq
and by \lem{taylor3} in \apd{taylor}, we get
\beq
\nm{ S_k(Dh) - (\exp{Dh}-I)D^{-1}h^{-1}} \le 1/(k+1)!\, .
\eeq

The error $\delta_j = \nm{\ket{y(jh)} - \ket{y_{j,0}}}$ can be bounded as follows. Note that $\delta_0=0$, and using the triangle inequality, we have
\begin{align}
\delta_{j+1} 
& = \nm{ \exp{Dh} \ket{ y(jh) }+ ( \exp{Dh} - I ) D^{-1} \ket{ c}  -   T_k (Dh) \ket{ y_{j,0} } - S_k (Dh) h \ket{ c } }\nonumber \\
& \leq \nm{ \exp{Dh} \ket{ y(jh) }   -   T_k (Dh) \ket{ y_{j,0} } } + \nm{( \exp{Dh} - I ) D^{-1} \ket{ c }- S_k (Dh) h \ket{ c }}\nonumber\\
& \leq \nm{ \exp{Dh} \ket{ y (jh) }   -  \exp{Dh} \ket{ y_{j,0} } } + \nm{ \exp{Dh} |y_{j,0} \rangle   -   T_k (Dh) | y_{j,0} \rangle } \nonumber\\
&\quad + \nm{( \exp{Dh} - I ) D^{-1} \ket{ c }- S_k (Dh) h \ket{ c } }\nonumber\\
& \leq \nm{\exp{Dh}} \nm{ \ket{ y (jh)} - | y_{j, 0} \rangle } + \nm{ \exp{Dh} - T_k (Dh) } \nm{ | y_{j, 0} \rangle }  \nonumber \\
&\quad  +\nm{ (\exp{Dh} - I ) D^{-1} - S_k (Dh)h}  \nm{ \ket{ c} }\nonumber \\
& \leq \delta_j + \frac 1{(k+1)!}\left[ \nm{ \ket{ y_{j,0} } }  + h \nm{ \ket{c} }\right].
\end{align}
This implies
\begin{equation}
\delta_j \leq  \frac j{(k+1)!}\left[ \max_{0 \le i \le j} \nm{\ket{y_{i,0}} } + h \nm{ \ket{c}}\right] . 
\label{errorbound}
\end{equation}

Next, we give an upper bound on $\max_{0\le i \le m} \nm{\ket{y_{i,0}} }$. Fix any $i \in \{0,1,\dots,m\}$. Then we have
\beq
\ket{y_{i,0}}=\bra{i(k+1)}C_{m,k,p}(D)^{-1} \ket{z},
\eeq
where the bra $\bra{i(k+1)}$ acts on the first register, and
\beq
\ket{z} = \ket{0}\ket{y_{\inn}} + h \dsum_{j=0}^{m-1} \ket{ j(k+1) +1} \ket{c}. 
\eeq
Therefore, by the triangle inequality, 
\beq
\nm{\ket{y_{i,0}}} 
\le
\nm{\bra{i(k+1)}C_{m,k,p}(D)^{-1} \ket{0}\ket{y_{\inn}}}
+h \dsum_{j=0}^{m-1} \nm{\bra{i(k+1)}C_{m,k,p}(D)^{-1}\ket{ j(k+1) +1} \ket{c}}.
\label{eq:yi0bound}
\eeq
To bound both terms, we consider the general expression $\nm{\bra{i(k+1)}C_{m,k,p}(D)^{-1}\ket{s} \ket{\beta}}$ for integer $s$.
Let $\ket{\beta}=\sum_{l=0}^{N-1} \beta_l \ket{l}$ for some $\beta_l \in \bbC$. Then since $C_{m,k,p}(D)=\sum_{l=0}^{N-1} C_{m,k,p}(\lambda_l) \otimes \ketbra{l}{l}$,
we get
\beq
\bra{i(k+1)}C_{m,k,p}(D)^{-1} \ket{s}\ket{\beta}=\dsum_{l=0}^{N-1} \beta_l
\bra{i(k+1)}C_{m,k,p}(\lambda_l)^{-1} \ket{s}\ket{l}.
\eeq
Then, by \lem{cond}, we have
\begin{align}
\nm{\bra{i(k+1)}C_{m,k,p}(D)^{-1} \ket{s}\ket{\beta}}^2 
&= \dsum_{l=0}^{N-1} \abs{\beta_l}^2
\abs{\bra{i(k+1)}C_{m,k,p}(\lambda_l)^{-1} \ket{s}}^2\nonumber \\
&\le 1.04e \dsum_{l=0}^{N-1} \abs{\beta_l}^2\nonumber \\
&= 1.04e \nm{\ket{\beta}}^2.
\end{align}
Hence we obtain
\begin{align}
\nm{\bra{i(k+1)}C_{m,k,p}(D)^{-1} \ket{0}\ket{y_{\inn}}}
&\le \sqrt{1.04 e} \nm{\ket{y_{\inn}}}, \\
\nm{\bra{i(k+1)}C_{m,k,p}(D)^{-1} \ket{j(k+1)+1}\ket{c}}
&\le \sqrt{1.04 e} \nm{\ket{c}}.
\end{align}
Using these two facts and the triangle inequality, Eq.~(\ref{eq:yi0bound}) implies
\beq
\nm{\ket{y_{i,0}}} \le \sqrt{1.04 e} \lb \nm{\ket{y_{\inn}}} + mh \nm{\ket{c}}\rb.
\eeq
Since this holds for any $i \in \{0,1,\dots,m\}$, we get
\beq
\max_{0 \le i \le m}\nm{\ket{y_{i,0}}} \le \sqrt{1.04 e} \lb \nm{\ket{y_{\inn}}} + mh \nm{\ket{c}}\rb.
\label{eq:yi0nm}
\eeq

Now using Eqs.~(\ref{errorbound}) and (\ref{eq:yi0nm}), we get
\begin{align}
\delta_j &\le (\sqrt{1.04 e}+1) j \lb \nm{\ket{y_{\inn}}} + mh \nm{\ket{c}}\rb/(k+1)!\nonumber \\
&= 
(\sqrt{1.04 e}+1) j \lb \nm{V^{-1}\ket{x_{\inn}}} + mh \nm{V^{-1}\ket{b}}\rb/(k+1)!\nonumber  \\
&\le 
(\sqrt{1.04 e}+1) j \nm{V^{-1}} \lb  \nm{\ket{x_{\inn}}} + mh \nm{\ket{b}}/(k+1)!\rb .
\label{eq:deltam}
\end{align}
Finally, recall that $\ket{x(t)}=V\ket{y(t)}$ and $\ket{x_{i,j}}=V\ket{y_{i,j}}$. Thus we have 
\begin{align}
\epsilon_j &= \nm{\ket{x(jh)}-\ket{x_{j,0}}}\nonumber \\
&= \nm{V(\ket{y(jh)}-\ket{y_{j,0}})}\nonumber \\
&\le \nm{V} \nm{\ket{y(jh)}-\ket{y_{j,0}}}\nonumber \\
&= \nm{V} \delta_j\nonumber \\
&\le (\sqrt{1.04 e}+1) j \nm{V} \nm{V^{-1}} \lb  \nm{\ket{x_{\inn}}} + mh \nm{\ket{b}}\rb/(k+1)! \nonumber \\
&\le 2.8 \kappa_V j \lb  \nm{\ket{x_{\inn}}} + mh \nm{\ket{b}}\rb/(k+1)!
\end{align}
for any $j \in \{0,1,\dots,m\}$, as claimed.
\end{proof}

\thm{solerr} implies that by choosing $(k+1)! \ge 3\kappa_V m \lb  \nm{\ket{x_{\inn}}} + mh \nm{\ket{b}}\rb/\epsilon$, we can ensure $\nm{\ket{x_{j,0}}-\ket{x(jh)}}\le \epsilon$ for any $j \in \{0,1,\dots,m\}$, so that $\ket{x_{j,0}}$ is close to $\ket{x(jh)}$ for any $j$. Furthermore, by \lem{statedist2} in \apd{state}, if $\nm{\ket{x_{j,0}}-\ket{x(jh)}}\le \epsilon$ and $\nm{\ket{x(jh)}} \ge \alpha$, then $\nm{\ket{x_{j,0}}/\nm{\ket{x_{j,0}}} - \ket{x(jh)}/\nm{\ket{x(jh)}}} \le 2\epsilon/\alpha$. So, provided that $\nm{\ket{x(jh)}}$ is large enough, the normalized state $\ket{x_{j,0}}/\nm{\ket{x_{j,0}}}$ is also close to the normalized state $\ket{x(jh)}/\nm{\ket{x(jh)}}$. 

%%%%%%%%%%%%%%%%%%%%%%%%%%%%%%%%%%%%%%%%%%%%%%%%%%%%%%%%%%%%%%%%%%%%%%%%%%%%%%
%%%%%%%%%%%%%%%%%%%%%%%%%%%%%%%%%%%%%%%%%%%%%%%%%%%%%%%%%%%%%%%%%%%%%%%%%%%%%%
\section{Success Probability}
\label{sec:succprob}

In the previous section, we have shown that $\ket{x_{m,0}}/\nm{\ket{x_{m,0}}}=\ket{x_{m,1}}/\nm{\ket{x_{m,1}}}=\dots =\ket{x_{m,p}}/\nm{\ket{x_{m,p}}}$ is a good approximation of $\ket{x(mh)}/\nm{\ket{x(mh)}}$, provided the truncation order $k$ is sufficiently large. In this section, we show that such a state can be obtained with non-negligible probability by measuring the first register of $\ket{x}/\nm{\ket{x}}$ in the standard basis, provided the padding parameter $p$ is sufficiently large. 

\begin{theorem}
Let $A=VD V^{-1}$ be a diagonalizable matrix, where $D=\diag{\lambda_0, \lambda_1,\dots,\lambda_{N-1}}$ satisfies $\Re(\lambda_i) \le 0$ for $i \in \{0,1,\dots,N-1\}$. Let $h \in \bbR^+$ such that $\nm{Ah} \le 1$. Let $\ket{{x}_{\inn}}, \ket{b} \in \bbC^{N}$, and let $\ket{x(t)}$ be defined by Eq.~(\ref{eq:dfeqsol2}).   Let $m, k, p \in \bbZ^+$ such that $(k+1)!\ge 70\kappa_V m (\nm{\ket{x_{\inn}}}+mh \nm{\ket{b}})/\nm{\ket{x(mh)}}$, where $\kappa_V=\nm{V}\cdot\nm{V^{-1}}$ is the condition number of $V$. Let $g={\max_{t \in [0,mh]} \nm{\ket{x(t)}}}/{\nm{\ket{x(mh)}}}$. Let $\ket{x}$ be defined by Eq.~(\ref{eq:lssol1}) and let $\ket{x_{i,j}}$ be defined by Eq.~(\ref{eq:lssol2}). Then for any $j \in \{0,1,\dots,p\}$, 
\beq
\dfrac{\nm{\ket{x_{m,j}}}}{\nm{\ket{x}}} \ge \dfrac{1}{\sqrt{p+77 mg^2}}.
\label{eq:succprob}
\eeq
\label{thm:succprob}
\end{theorem}

\begin{proof}
Recall that $\ket{x_{m,j}}=\ket{x_{m,0}}$ for all $j\in\{1,2,\dots,p\}$. Thus it is sufficient to prove Eq.~(\ref{eq:succprob}) for $j=0$.

Define
\beq
\ket{x_{\rm good}} \defeq \dsum_{j=0}^p \ket{m(k+1)+j} \ket{x_{m,j}}
= \lb \dsum_{j=0}^p \ket{m(k+1)+j} \rb \ket{x_{m,0}}
\eeq
and
\beq
\ket{x_{\rm bad}} \defeq \dsum_{i=0}^{m-1}\dsum_{j=0}^k \ket{i(k+1)+j} \ket{x_{i,j}}.
\eeq
Then $\ket{x}=\ket{x_{\rm good}}+\ket{x_{\rm bad}}$ and $\braket{x_{\rm good}}{x_{\rm bad}} = 0$, so 
\begin{align}
\nm{\ket{x}}^2 &= \nm{\ket{x_{\rm good}}}^2+\nm{\ket{x_{\rm bad}}}^2\nonumber \\
&= (p+1) \nm{\ket{x_{m,0}}}^2 + \nm{\ket{x_{\rm bad}}}^2.
\label{eq:nmx}
\end{align}

Next we give a lower bound on $\nm{\ket{x_{m,0}}}$ and an upper bound on $\nm{\ket{x_{\rm bad}}}$. Let $q=\nm{\ket{x(mh)}}$. Then by the definition of $g$, we have $\nm{\ket{x(ih)}} \le gq$ for $i \in \{0,1,\dots,m-1\}$. Meanwhile, by \thm{solerr} and the choice of $k$, we have 
\beq
\nm{\ket{x_{i,0}}-\ket{x(ih)}} \le 0.04 q, \qquad 0 \le i \le m.
\eeq
As a result, we get
\beq
\nm{\ket{x_{i,0}}} \le (g+0.04) q \le 1.04 gq, \qquad 0 \le i \le m-1,
\label{eq:nmxi0}
\eeq
and
\beq
0.96 q \le  \nm{\ket{x_{m,0}}} \le 1.04 q.
\label{eq:nmxm0}
\eeq

Now for any $i \in \{0,1,\dots,m-1\}$, since $\ket{x_{i,j}} = ({Ah}/{j}) \ket{x_{i,j-1}}$, for $j \in \{2,3, \dots,k\}$, we get
\beq
\ket{x_{i,j}} = \dfrac{(Ah)^{j-1}}{j!} \ket{x_{i,1}}, \qquad 2 \le j \le k.
\eeq
Then, since $\nm{Ah} \le 1$, we get 
\beq
\nm{\ket{x_{i,j}}} \le \dfrac{\nm{\ket{x_{i,1}}}}{j!}, \qquad 2 \le j \le k.
\label{eq:nmdecay}
\eeq

Next, using the fact $\ket{x_{i+1,0}} = \ket{x_{i,0}}+\sum_{j=1}^k \ket{x_{i,j}}$ and the triangle inequality, we get
\begin{align}
2.08 g q
& \ge  
\nm{\ket{x_{i+1,0}}} + \nm{\ket{x_{i,0}}}\nonumber \\
& \ge 
\nm{\ket{x_{i+1,0}}-\ket{x_{i,0}}}\nonumber \\
&\ge  
\nm{\ket{x_{i,1}}} - \dsum_{j=2}^k \nm{\ket{x_{i,j}}}\nonumber \\
&\ge 
\lb 1 - \dsum_{j=2}^k \dfrac{1}{j!}\rb \nm{\ket{x_{i,1}}}\nonumber \\
& \ge 
(3-e) \nm{\ket{x_{i,1}}},
\end{align}
which implies
\beq
\nm{\ket{x_{i,1}}} \le \dfrac{2.08 g q}{3-e}, \qquad 0 \le i \le m-1.
\label{eq:nmxi1}
\eeq
Then it follows from Eq.~(\ref{eq:nmdecay}) that
\beq
\nm{\ket{x_{i,j}}} \le  \dfrac{2.08 g q}{j!(3-e)}, \qquad 0 \le i \le m-1, ~1 \le j \le k.
\label{eq:nmxij}
\eeq

Now using Eqs. \eqref{eq:nmxi0}, \eqref{eq:nmxi1}, and \eqref{eq:nmxij},
we obtain
\begin{align}
\nm{\ket{x_{\rm bad}}}^2
&=\dsum_{i=0}^{m-1}  \nm{\ket{x_{i,0}}}^2
+\dsum_{i=0}^{m-1} \dsum_{j=1}^k \nm{\ket{x_{i,j}}}^2\nonumber\\
&\le 1.04^2 m g^2 q^2 +
m \dsum_{j=1}^k \dfrac{( 2.08 gq)^2}{(j!)^2(3-e)^2}\nonumber \\
&\le 70.9 mg^2 q^2,
\label{eq:nmxbad}
\end{align}
where the last step follows from
\beq
\dsum_{j=1}^k\dfrac{1}{(j!)^2} \le I_0(2)-1 < 1.28
\eeq
as in Eq.~(\ref{eq:factsquaresum}).
Thus, combining Eqs.~\eqref{eq:nmx}, \eqref{eq:nmxm0}, and \eqref{eq:nmxbad} we get
\begin{align}
\dfrac{\nm{\ket{x_{m,0}}}^2}{\nm{\ket{x}}^2} &\ge
\dfrac{(0.96q)^2}{p(0.96q)^2+70.9 mg^2q^2} \nonumber\\
& \ge  \dfrac{1}{p+77 mg^2},
\end{align}
as claimed.
\end{proof}

\thm{succprob} implies that by choosing $p=m$, we can make the probability at least $1/78g^2$ for obtaining the state $\ket{x_{m,j}}/\nm{\ket{x_{m,j}}}$ when measuring the first register of $\ket{x}/\nm{\ket{x}}$ in the standard basis. This probability can be increased to $\omg{1}$ by amplitude amplification, which uses $\bgo{g}$ repetitions of the above procedure.

%%%%%%%%%%%%%%%%%%%%%%%%%%%%%%%%%%%%%%%%%%%%%%%%%%%%%%%%%%%%%%%%%%%%%%%%%%%%%%
%%%%%%%%%%%%%%%%%%%%%%%%%%%%%%%%%%%%%%%%%%%%%%%%%%%%%%%%%%%%%%%%%%%%%%%%%%%%%%
\section{State Preparation}
\label{sec:stateprep}

To apply the QLSA to a linear system of the form $M|x \rangle = | y \rangle$, we must also be able to prepare the state $| y \rangle$. To quantify the complexity of this subroutine, Ref.~\cite{berry2014high} assumes that $\vec x_{\inn}$ and $\vec b$ are sparse vectors whose entries are given by oracles. Instead of assuming sparsity, here we simply assume that we have controlled oracles that produce states proportional to $\vec x_{\inn}$ and $\vec b$, respectively.

The following lemma shows how to use these oracles to produce the state appearing on the right-hand side of our linear system. We write $\ket{\bar{\phi}}$ to denote the normalized version of $\ket{\phi}$, i.e., $\ket{\bar{\phi}} \defeq \ket{\phi}/\nm{\ket{\phi}}$, for any $\ket{\phi}$.

\begin{lemma} 
Let $\mathcal{O}_x$ be a unitary that maps $\ket{1}\ket{\phi}$ to $\ket{1}\ket{\phi}$ for any $\ket{\phi}$ and maps $\ket{0}\ket{0}$ to $\ket{0}\ket{\bar{x}_{\inn}}$, where $\bar{x}_{\inn}=\vec x_{\inn}/\nm{\vec x_{\inn}}$. Let $\mathcal{O}_b$ be a unitary that maps $\ket{0}\ket{\phi}$ to $\ket{0}\ket{\phi}$ for any $\ket{\phi}$ and maps $\ket{1}\ket{0}$ to $\ket{1}\ket{\bar{b}}$, where $\bar{b}=\vec b / \bnm{\vec b}$
. Suppose we know $\nm{\vec x_{\inn}}$ and $\bnm{\vec b}$. Then the state proportional to
\begin{equation}
\ket{0} \ket{x_{\inn}} + h \displaystyle\sum_{i=0}^{m-1} | i(k+1) +1 \rangle  \ket{ b}
\end{equation}
can be produced with a constant number of calls to $\mathcal{O}_x$ and $\mathcal{O}_b$, and $\poly{\log{mk}}$ elementary gates.
\label{lem:stateprepare}
\end{lemma}
\begin{proof}

Consider the initial state $|0 \rangle | 0 \rangle$, where the first register is the $(d+1)$-dimensional register corresponding to the block-level indexing of $C_{m,k,p}$, and the second register is the $N$-dimensional register that stores $\vec b$ and $\vec x_{\inn}$. We perform the unitary
\begin{align}
U & = \left( \frac{\nm{ \vec x_{\inn}}}{\sqrt{\nm{\vec x_{\inn}}^2+m h^2 \bnm{\vec b}^2}} |0 \rangle + \frac{\sqrt{m} h \bnm{ \vec b}}{\sqrt{\nm{\vec x_{\inn}}^2+m h^2 \bnm{\vec b}^2}} | 1 \rangle \right) \langle 0 | \nonumber \\
&\quad  + \left( \frac{\nm{ \vec x_{\inn}}}{\sqrt{\nm{\vec x_{\inn}}^2+m h^2 \bnm{\vec b}^2}} |1 \rangle - \frac{\sqrt{m} h \bnm{ \vec b}}{\sqrt{\nm{\vec x_{\inn}}^2+m h^2 \bnm{\vec b}^2}} | 0 \rangle \right) \langle 1 | \nonumber \\
&\quad + \sum_{j=2}^{d} |j \rangle \langle j| 
\end{align}
on the first register to get the state 
\begin{equation}
| \psi \rangle = \frac{1}{\sqrt{\nm{\vec x_{\inn}}^2+m h^2 \bnm{\vec b}^2}} \lb
{\nm{ \vec x_{\inn}}}|0 \rangle + {\sqrt{m} h \nm{ \vec b} } | 1 \rangle \rb | 0 \rangle .
\end{equation}

Next, we apply the unitaries $\mathcal{O}_x$ and $\mathcal{O}_b$ (in arbitrary order), and obtain the state
\begin{align}
\ket{\psi'}
&=  \dfrac{1}{\sqrt{\nm{\vec x_{\inn}}^2+m h^2 \bnm{\vec b}^2}} \lb 
\nm{\vec x_{\inn}} |0 \rangle \ket{\bar{x}_{\inn}} + {\sqrt{m} h \bnm{\vec b} } \ket{1} \ket{\bar{b}} \rb \nonumber\\
&=  \dfrac{1}{\sqrt{\nm{\vec x_{\inn}}^2+m h^2 \bnm{\vec b}^2}} \lb \ket{0}\ket{x_{\inn}} + {\sqrt{m} h } \ket{1} \ket{b} \rb.
\end{align}

Finally, we apply a unitary that maps $\ket{0}$ to $\ket{0}$ and maps $\ket{1}$ to $\frac{1}{\sqrt{m}}\sum_{j=0}^{m-1} | j(k+1) + 1 \rangle $ on the first register to get the state we need.
This can be done by standard techniques using $\poly{\log{mk}}$ elementary gates.
\end{proof}

%%%%%%%%%%%%%%%%%%%%%%%%%%%%%%%%%%%%%%%%%%%%%%%%%%%%%%%%%%%%%%%%%%%%%%%%%%%%%%
%%%%%%%%%%%%%%%%%%%%%%%%%%%%%%%%%%%%%%%%%%%%%%%%%%%%%%%%%%%%%%%%%%%%%%%%%%%%%%
\section{Main Result}
\label{sec:main}

In this section, we formally state and prove our main result on the quantum algorithm for linear ordinary differential equations.

\begin{theorem}
Suppose $A=VDV^{-1}$ is an $N \times N$ diagonalizable matrix, where $D=\diag{\lambda_0,\lambda_1,\dots,\lambda_{N-1}}$ satisfies $\Re(\lambda_j) \le 0$ for any $j \in \{0,1,\dots,N-1\}$. In addition, suppose $A$ has at most $s$ nonzero entries in any row and column, and we have an oracle $\mathcal{O}_A$ that computes these entries. Suppose $\vec{x}_{\inn}$ and $\vec b$ are $N$-dimensional vectors with known norms and that we have two controlled oracles, $\mathcal{O}_x$ and $\mathcal{O}_b$, that prepare the states proportional to $\vec{x}_{\inn}$ and $\vec b$, respectively. Let $\vec x$ evolve according to the differential equation
\beq
\dfrac{d \vec x}{d t} = A \vec x + \vec b
\eeq
with the initial condition $\vec x(0)=\vec x_{\inn}$. Let $T>0$ and 
\beq
g \defeq \max_{t \in [0,T]}\nm{\vec x(t)}/ \nm{\vec x(T)} .\label{decaybound} 
\eeq
Then there exists
a quantum algorithm that produces a state $\epsilon$-close to $\vec x(T)/\nm{\vec x(T)}$ in $\ell^2$ norm, succeeding with probability $\omg{1}$, with a flag indicating success, using 
\beq
\bgo{\kappa_V sg T \nm{A} \cdot \poly{\log{\kappa_V s g \beta T \nm{A}/\epsilon}}}
\eeq
 queries to $\mathcal{O}_A$, $\mathcal{O}_x$, and $\mathcal{O}_b$, where 
$\kappa_V=\nm{V}\cdot \nm{V^{-1}}$ is the condition number of $V$ and $\beta = ( \nm{\ket{x_{\inn}}}+T \nm{\ket{b}} )/ \nm{\ket{x(T)}}$. The gate complexity of this algorithm is  larger than its query complexity by a factor of $\poly{\log{\kappa_V s g \beta T \nm{A} N  /\epsilon}}$.

\label{thm:main}
\end{theorem}

\begin{proof}
Recall that we use $\ket{\bar{\phi}}$ to denote the normalized version of $\ket{\phi}$.

%-----------------------------------------------------------------------------
\paragraph{Statement of the Algorithm} Let $h=T/\lceil T\nm{A} \rceil$, $m=p=T/h=\lceil T\nm{A}\rceil$, $\delta = \epsilon/(25\sqrt{m}g)$, $\epsilon\le 1/2$, and 
\beq
k=\left\lfloor \frac{2\log{\Omega}}{\log{\log{\Omega}}} \right\rfloor
\eeq
with
\begin{align}
\Omega &= 2.8\kappa_V m (\nm{\ket{x_{\inn}}}+mh\nm{\ket{b}})/(\delta\nm{\ket{x(T)}})  \nonumber \\
&=70g\kappa_V m^{3/2} (\nm{\ket{x_{\inn}}}+T\nm{\ket{b}})/(\epsilon \nm{\ket{x(T)}}) .
\end{align}
This choice of $k$ ensures that $(k+1)!\ge\Omega$.
\begin{comment}
To show this result, take $\zeta=\log{\log{\Omega}}$, so
\beq
k+1 \ge \frac{2 e^\zeta}{\zeta}.
\eeq
Then, using Stirling's approximation,
\begin{align}
\log{(k+1)!} &\ge (k+1)[\log{k+1}-1] + \frac 12 \log{2\pi(k+1)} \nonumber \\
&\ge \frac{2 e^\zeta}{\zeta}\left[ \log{\frac{2 e^\zeta}{\zeta}}-1\right] + \frac 12 \log{\frac{4\pi e^\zeta}{\zeta}} \nonumber \\
&\ge \frac{2 \log{\Omega}}{\zeta}\left[ \log{\frac{2 e^\zeta}{\zeta}}-1\right].
\end{align}
Then note that
\beq
\frac{d}{d\zeta}\frac{2}{\zeta}\left[ \log{\frac{2 e^\zeta}{\zeta}}-1\right]=\frac{2\log{\zeta/2}}{\zeta^2},
\eeq
which is zero for $\zeta=2$, in which case
\beq
\frac{2}{\zeta}\left[ \log{\frac{2 e^\zeta}{\zeta}}-1\right]=1.
\eeq
It is easily seen that this is a minimum, which implies that
\beq
\frac{2}{\zeta}\left[ \log{\frac{2 e^\zeta}{\zeta}}-1\right]\ge 1,
\eeq
so $\log{(k+1)!} \ge \log{\Omega}$ and therefore $(k+1)!\ge\Omega$ as claimed.
\end{comment}
Moreover, since $\Omega\ge 70$, $k\ge 5$.
Hence this is an appropriate choice of $k$ for the conditions of \thm{succprob}.
We build the linear system
\beq
C_{m, k, p} (Ah) \ket{x} = \ket{z} \defeq \ket{0} \ket{x_{\inn}} + h \displaystyle\sum_{i=0}^{m-1} | i(k+1) +1 \rangle  \ket{ b}.
\label{eq:ls2}
\eeq
We use the QLSA from Ref.~\cite{childs2015quantum} to solve this linear system and obtain a state $\ket{\bar{x}'}$ such that $\nm{\ket{\bar{x}}-\ket{\bar{x}'}}\le \delta$. Then we measure the first register of $\ket{\bar{x}'}$ in the standard basis, and conditioned on the outcome being in
\beq
S \defeq \{m(k+1), m(k+1)+1, \dots, m(k+1)+p\},
\eeq
we output the state of the second register. We will show that the probability of this event happening is $\omg{1/g^2}$. Using amplitude amplification \cite{brassard2002quantum}, we can raise this probability to $\omg{1}$ with $\bgo{g}$ repetitions of the above procedure.

%-----------------------------------------------------------------------------
\paragraph{Proof of Correctness} Let $d=m(k+1)+p$. Let $\ket{x_{i,j}}$ be defined by Eq.~(\ref{eq:lssol2}), and let $\ket{x_l}=\ket{x_{i,j}}$ for $l=i(k+1)+j$. Then we have
\beq
\ket{x} = \dsum_{l=0}^d  \ket{l} \ket{x_{l}}.
\eeq
Note that $\ket{x_{l}}=\ket{x_{m,0}}$ for any $l \in S$. Then by \thm{solerr} and our choice of parameters, we have
for any $l \in S$, 
\beq
\nm{\ket{x_l} - \ket{x(T)}} \le \delta \nm{\ket{x(T)}}.
\label{eq:distxlxt}
\eeq

Now let $\alpha_l$ be such that
\beq
\ket{\bar{x}} = \dsum_{l=0}^d \alpha_{l} \ket{l} \ket{\bar{x}_{l}}.
\eeq
Then for any $l \in S$, we have (by \thm{succprob} and our choice of parameters)
\beq
\alpha_l = \dfrac{\nm{\ket{x_l}} }{\nm{\ket{x}}} \ge \dfrac{1}{\sqrt{78m}g}.
\eeq
By Eq.~(\ref{eq:distxlxt}) and \lem{statedist2} in \apd{state}, we have
\beq
\nm{\ket{\bar{x}_l}-\ket{\bar{x}(T)}} \le 2\delta.
\label{eq:distbarxlxt}
\eeq

Now suppose the QLSA outputs the state
\beq
\ket{\bar{x}'} = \dsum_{l} \alpha_{l}' \ket{l} \ket{\bar{x}_{l}'}
\eeq
which satisfies
\beq
\nm{\ket{\bar{x}}-\ket{\bar{x}'}}\le \delta.
\eeq
Then for any $l \in S$, by \lem{statedist} in \apd{state}, we have
\beq
\nm{\ket{\bar{x}_l} - \ket{\bar{x}_l'}} \le \dfrac{2\delta}{\alpha_l - \delta}.
\label{eq:distbarxlxl}
\eeq
Then it follows from Eqs.~(\ref{eq:distbarxlxt}) and (\ref{eq:distbarxlxl}) that for any $l \in S$, 
\begin{align}
\nm{\ket{\bar{x}_l'} - \ket{\bar{x}(T)}} &\le
\nm{\ket{\bar{x}_l} - \ket{\bar{x}(T)}}+
\nm{\ket{\bar{x}_l} - \ket{\bar{x}_l'}} \nonumber\\
& \le 2\delta\left(1+\dfrac{1}{\alpha_l - \delta}\right) \le \epsilon.
\end{align}
Furthermore, by \lem{probdist} in \apd{state}, we have
\beq
\alpha_l' \ge \alpha_l -\delta \ge \dfrac{1}{11\sqrt{m} g}.
\eeq
Therefore, if we measure the first register of $\ket{\bar{x}'}$ in the standard basis, the probability of getting some outcome $l \in S$ is
\beq
\dsum_{l \in S}\abs{\alpha_l'}^2 \ge \dfrac{p}{121 m g^2} = \dfrac{1}{121g^2},
\eeq
and when this happens the state of the second register becomes $\ket{\bar{x}_l'}$ for some $l \in S$, which is $\epsilon$-close to the desired $\ket{\bar{x}(T)}$ in $\ell^2$ norm. The success probability can be raised to $\omg{1}$ by using $O(g)$ rounds of amplitude amplification.

%-----------------------------------------------------------------------------
\paragraph{Analysis of the Complexity} The matrix $C_{m,k,p}(A)$ is a $(d+1)N \times (d+1)N$ matrix with $\bgo{ks}$ nonzero entries in any row or column. By \thm{cond} and our choice of parameters, the condition number of $C_{m,k,p}(A)$ is $\bgo{\kappa_V km}$. Consequently, by Theorem 5 of Ref.~\cite{childs2015quantum}, the QLSA produces the state $\ket{x'}$ with
\beq
\bgo{\kappa_V k^2 m s \cdot \poly{\log{\kappa_V k m s/\delta}}}
=
\bgo{\kappa_V  s T \nm{A} \cdot \poly{\log{\kappa_V s g \beta T \nm{A} /\epsilon}}}
\eeq
queries to the oracles $\mathcal{O}_A$, $\mathcal{O}_{x}$, and $\mathcal{O}_{b}$, where $\beta = \lb \nm{\ket{x_{\inn}}+T \nm{\ket{b}}} \rb/ \nm{\ket{x(T)}}$, and its gate complexity is larger by a factor of $\poly{\log{\kappa_V k m sN/\delta}}=\poly{\log{\kappa_V s g \beta T \nm{A} N  /\epsilon}}$. Since amplitude amplification requires only $\bgo{g}$ repetitions of this procedure, the query complexity of our algorithm is $\bgo{\kappa_V  s g T \nm{A}\cdot\poly{\log{\kappa_V s g \beta T \nm{A} /\epsilon}}}$ and its gate complexity is larger by a factor of $\poly{\log{\kappa_V s g \beta T \nm{A} N  /\epsilon}}$, as claimed.
\end{proof}

%%%%%%%%%%%%%%%%%%%%%%%%%%%%%%%%%%%%%%%%%%%%%%%%%%%%%%%%%%%%%%%%%%%%%%%%%%%%%%
%%%%%%%%%%%%%%%%%%%%%%%%%%%%%%%%%%%%%%%%%%%%%%%%%%%%%%%%%%%%%%%%%%%%%%%%%%%%%%
\section{Discussion}
\label{sec:discussion}

In this paper, we have presented an algorithm for solving (possibly inhomogeneous) linear ordinary differential equations with constant coefficients, with exponentially improved performance over the algorithm of Ref.~\cite{berry2014high}.

The complexity of our algorithm depends on the parameter $g$ defined in Eq.~\eqref{decaybound}, which characterizes the extent to which the final solution vector decays relative to the solution vector at earlier times.  The success probability obtained from a single solution of the linear system (as analyzed in \sec{succprob}) decays with $g$, so we boost the success probability using amplitude amplification.  Dramatically improved dependence on $g$ is unlikely since a simulation of evolution in which the state decays can be used to implement postselection, following a procedure like that of Proposition 5 of Ref.~\cite{aaronson2005quantum}.  In particular, the ability to postselect on an exponentially small amplitude would imply $\BQP=\PP$, which is considered implausible.  Note that the algorithm of Ref.~\cite{berry2014high} has essentially the same limitation (although there it was stated in terms of an assumed upper bound on $g$).

Aside from its exponentially improved dependence on $\epsilon$, our approach has other advantages over Ref.~\cite{berry2014high}.  First, since linear multistep methods may be unstable, Ref.~\cite{berry2014high} requires careful selection of a stable method.  In contrast, the propagator $\exp{At}$ exactly evolves the state forward in time, so we need not concern ourselves with numerical stability. Furthermore, our algorithm has nearly linear scaling with respect to the evolution time $T$, which is a quadratic improvement over Ref.~\cite{berry2014high}. Since Hamiltonian simulation is a special case of our algorithm with $g=1$, the no-fast-forwarding theorem \cite{BACS07} implies that the dependence of our algorithm on $T$ is optimal up to logarithmic factors. We also obtain polynomial improvements over Ref.~\cite{berry2014high} for the dependence of the complexity on the sparsity $s$, norm $\nm{A}$, and condition number $\kappa_V$ of the transformation that diagonalizes $A$.

Although \thm{main} assumes that $A$ is diagonalizable, our algorithm can also produce approximate solutions for non-diagonalizable $A$. This is because diagonalizable matrices are dense within the set of all complex matrices: for any non-diagonalizable matrix $A$ and any $\delta>0$, there is a diagonalizable matrix $B$ such that $\nm{A-B} < \delta$.  Using such a $B$ in place of $A$, we can simulate Eq.~(\ref{eq:dfeq}) approximately.
{This approach can yield a matrix $B$ whose diagonalizing transformation has a condition number polynomial in $1/\delta$, so the complexity would no longer be $\poly{\log{1/\epsilon}}$.}

Our work raises some natural open problems.  For simplicity, we have assumed that the matrix $A$ and the inhomogeneity $\vec b$ are independent of $t$.  More generally, can one solve differential equations with time-dependent coefficients with complexity $\poly{\log{1/\epsilon}}$?  This is possible in the case of quantum simulation, i.e., when $A(t)$ is anti-Hermitian and $\vec b=0$ \cite{berry2014exponential}.  However, some aspects of our analysis appear to fail in the time-dependent case.

Finally, while our algorithm has optimal dependence on $\epsilon$ and nearly optimal dependence on $T$, the joint dependence on these parameters might be improved.  Recent work gave an algorithm for Hamiltonian simulation with complexity $\smash{O(t+\frac{\log{1/\epsilon}}{\loglog{1/\epsilon}})}$, providing an optimal tradeoff.  Can similar complexity be attained for more general differential equations?

%%%%%%%%%%%%%%%%%%%%%%%%%%%%%%%%%%%%%%%%%%%%%%%%%%%%%%%%%%%%%%%%%%%%%%%%%%%%%%
%%%%%%%%%%%%%%%%%%%%%%%%%%%%%%%%%%%%%%%%%%%%%%%%%%%%%%%%%%%%%%%%%%%%%%%%%%%%%%
\section*{Acknowledgements}

This research was supported by the Canadian Institute for Advanced Research, the National Science Foundation (grant number 1526380), and
IARPA (contract number D15PC00242).
In addition, DWB is funded by an Australian Research Council Discovery Project (DP160102426).

%%%%%%%%%%%%%%%%%%%%%%%%%%%%%%%%%%%%%%%%%%%%%%%%%%%%%%%%%%%%%%%%%%%%%%%%%%%%%%
%%%%%%%%%%%%%%%%%%%%%%%%%%%%%%%%%%%%%%%%%%%%%%%%%%%%%%%%%%%%%%%%%%%%%%%%%%%%%%

\providecommand{\bysame}{\leavevmode\hbox to3em{\hrulefill}\thinspace}

\begin{appendix}

%%%%%%%%%%%%%%%%%%%%%%%%%%%%%%%%%%%%%%%%%%%%%%%%%%%%%%%%%%%%%%%%%%%%%%%%%%%%%%
%%%%%%%%%%%%%%%%%%%%%%%%%%%%%%%%%%%%%%%%%%%%%%%%%%%%%%%%%%%%%%%%%%%%%%%%%%%%%%
\section{Lemmas About Taylor Series}

In this appendix, we establish some basic lemmas about the approximation of functions by truncated Taylor series.

\label{apd:taylor}
\begin{lemma}
Let $k \in \bbZ^+$, $z \in \bbC$, $\abs{z} \le 1$, and $\Re(z) \le 0$. Define $T_k(z)\defeq \sum_{j=0}^k \frac{z^j}{j!}$. Then
$\abs{T_k(z) - \exp{z}} \le 1/(k+1)!$ and $\abs{T_k(z)} \le 1+1/(k+1)!$.
\label{lem:taylor1}
\end{lemma}

\begin{proof} Given $\phi$ such that $z=e^{i\phi}|z|$, define a function $f(x):=\exp{e^{i\phi} x}$ for $x\in\bbR$.
Then the integral form of the remainder for Taylor's theorem gives
\beq
f(x) = T_k(e^{i\phi} x) + \frac 1{k!} \int_{0}^x (x-t)^{k}  f^{(k+1)}(t) dt,
\eeq
so
\beq
\label{remainder}
f(x)-T_k(e^{i\phi} x) = \frac {e^{i(k+1)\phi}}{k!} \int_{0}^x (x-t)^{k} \exp{e^{i\phi} t} dt.
\eeq
Because $\Re(z) \le 0$, $|\exp{e^{i\phi} t}|\le 1$.
Taking $x=|z|$, we therefore have
\beq
\abs{T_k(z) - \exp{z}} 
= \frac 1{k!}\abs{\int_{0}^{|z|} (|z|-t)^{k} \exp{e^{i\phi} t} dt}\le \frac 1{k!}\int_{0}^{|z|} (|z|-t)^{k} dt = \frac {|z|^{k+1}}{(k+1)!} \le \frac 1{(k+1)!}.
\eeq
Using this bound,
\beq
\abs{T_k(z)} \le \abs{\exp{z}} + \abs{\exp{z}-T_k(z)}
\le 1+\frac 1{(k+1)!}
\eeq
and the lemma follows.
\end{proof}

\begin{lemma}
Let $b \in \bbZ^*$, $k \in \bbZ^+$, $k\ge 5$, and $b \le k$. Let $z \in \bbC$, $\abs{z} \le 1$, and $\Re(z) \le 0$. Define $T_{b,k}(z) \defeq \sum_{j=b}^k \frac{b! z^{j-b}}{j!}$. Then $\abs{T_{b,k}(z)}\le \sqrt{1.04}$.
\label{lem:taylor2}
\end{lemma}

\begin{proof}
Again we can take $z=e^{i\phi}|z|$ and $f(x):=\exp{e^{i\phi} x}$ for $x\in\bbR$.
If $b=0$, then $T_{b,k}(z)=T_k(z)$, so from \lem{taylor1} we have $\abs{T_{b,k}(z)} \le 1+1/(k+1)!<\sqrt{1.04}$.
If $b=k$, then $\abs{T_{b,k}(z)}=1$.
If $b=k-1$, then
\beq
T_{b,k}(z) = 1+\frac z{k}.
\eeq
Then we obtain
\beq
\abs{T_{b,k}(z)} = \sqrt{1+\frac{2\Re(z)}{k} + \frac{|z|^2}{k^2}} \le \sqrt{1+\frac 1{k^2}},
\eeq
since $\Re(z) \le 0$ and $\abs{z} \le 1$.
Then, for $k\ge 5$, we obtain $\abs{T_{b,k}(z)}\le \sqrt{1.04}$.

Otherwise, $1 \le b<k-1$, and using the integral form of the remainder gives
\beq
\sum_{j=b}^{\infty} \frac{z^{j}}{j!} = \frac {e^{ib\phi}}{(b-1)!} \int_{0}^{|z|} (|z|-t)^{b-1} \exp{e^{i\phi} t} dt .
\eeq
Then limiting the sum to $k$ gives
\beq
\sum_{j=b}^{k} \frac{z^{j}}{j!} = \frac {e^{ib\phi}}{(b-1)!} \int_{0}^{|z|} (|z|-t)^{b-1} \exp{e^{i\phi} t} dt - \frac {e^{i(k+1)\phi}}{k!} \int_{0}^{|z|} (|z|-t)^{k} \exp{e^{i\phi} t} dt .
\eeq
Therefore, we obtain a formula for $T_{b,k}(z)$ as
\beq
T_{b,k}(z) = \frac {b}{|z|^b} \int_{0}^{|z|} (|z|-t)^{b-1} \exp{e^{i\phi} t} dt - \frac {b!e^{i(k-b+1)\phi}}{k!|z|^b} \int_{0}^{|z|} (|z|-t)^{k} \exp{e^{i\phi} t} dt .
\eeq
Again $|\exp{e^{i\phi} t}|\le 1$, so taking the absolute value gives us
\begin{align}
\abs{T_{b,k}(z)} &\le \frac {b}{|z|^b} \int_{0}^{|z|} (|z|-t)^{b-1} dt + \frac {b!}{k!|z|^b} \int_{0}^{|z|} (|z|-t)^{k} dt \nonumber \\
&= 1+ \frac{b! |z|^{k-b+1}}{(k+1)!}\nonumber \\
&\le 1+ \frac{b!}{(k+1)!}\nonumber \\
&\le 1+ \frac{1}{(k-1)k(k+1)}\nonumber \\
&\le 1+\frac 1{120},
\end{align}
where in the last line we have used $k\ge 5$.
Hence for $0 \le b\le k$, we obtain $\abs{T_{b,k}(z)}\le \sqrt{1.04}$ as required.
\end{proof}

\begin{lemma} Let $k \in \bbZ^+$. Let $z \in \bbC$, $|z| \leq 1$. Define 
$S_k(z) \defeq \sum_{j=1}^k \frac{z^{j-1}}{j!}$. Then $\abs{S_{k}(z) - (\exp{z} - 1)z^{-1} } \leq 1/(k+1)!$.
\label{lem:taylor3}
\end{lemma}

\begin{proof}
Again we can take $z=e^{i\phi}|z|$ and $f(x):=\exp{e^{i\phi} x}$ for $x\in\bbR$.
Then using Eq.~\eqref{remainder},
\begin{align}
\frac {e^{i(k+1)\phi}}{k!} \int_{0}^x (x-t)^{k} \exp{e^{i\phi} t} dt &= f(x)-T_k(e^{i\phi} x) \nonumber \\
&= (f(x)-1)-(T_k(e^{i\phi} x)-1)\nonumber \\
&= \exp{e^{i\phi} x}-1-e^{i\phi} x S_k(e^{i\phi} x).
\end{align}
Dividing by $e^{i\phi} x$ gives
\beq
\frac{\exp{e^{i\phi} x}-1}{e^{i\phi} x}- S_k(e^{i\phi} x)=\frac {e^{ik\phi}}{k!x} \int_{0}^x (x-t)^{k} \exp{e^{i\phi} t} dt.
\eeq
Taking the absolute value then gives
\beq
\abs{S_{k}(z) - (\exp{z} - 1)z^{-1} } \leq \frac {1}{k!|z|} \int_{0}^{|z|} (|z|-t)^{k} dt = \frac{|z|^k}{(k+1)!} \le \frac 1{(k+1)!}
\eeq
as claimed.
\end{proof}

%%%%%%%%%%%%%%%%%%%%%%%%%%%%%%%%%%%%%%%%%%%%%%%%%%%%%%%%%%%%%%%%%%%%%%%%%%%%%%
%%%%%%%%%%%%%%%%%%%%%%%%%%%%%%%%%%%%%%%%%%%%%%%%%%%%%%%%%%%%%%%%%%%%%%%%%%%%%%
\section{Lemmas About Quantum States}
\label{apd:state}

In this appendix, we prove some technical lemmas about approximation of quantum states.

\begin{lemma}
Let $\ket{\psi}$ and $\ket{\phi}$ be two vectors such that $\nm{\ket{\psi}} \ge \alpha >0$ and $\nm{\ket{\psi}-\ket{\phi}} \le \beta$. Then \beq
\nm{\dfrac{\ket{\psi}}{\nm{\ket{\psi}}}-\dfrac{\ket{\phi}}{\nm{\ket{\phi}}}} \le \dfrac{2\beta}{\alpha}.
\eeq
\label{lem:statedist2}
\end{lemma}

\begin{proof}
Using the triangle inequality, we get
\begin{align}
\nm{\dfrac{\ket{\psi}}{\nm{\ket{\psi}}}-\dfrac{\ket{\phi}}{\nm{\ket{\phi}}}}
&=
\nm{\dfrac{\ket{\psi}}{\nm{\ket{\psi}}}-\dfrac{\ket{\phi}}{\nm{\ket{\psi}}}+\dfrac{\ket{\phi}}{\nm{\ket{\psi}}}-
\dfrac{\ket{\phi}}{\nm{\ket{\phi}}}} \nonumber\\
&\le
\nm{\dfrac{\ket{\psi}}{\nm{\ket{\psi}}}-\dfrac{\ket{\phi}}{\nm{\ket{\psi}}}} + 
\nm{\dfrac{\ket{\phi}}{\nm{\ket{\psi}}}-
\dfrac{\ket{\phi}}{\nm{\ket{\phi}}}} \nonumber \\
& \le
\dfrac{\nm{\ket{\psi}-\ket{\phi}}}{\nm{\ket{\psi}}} + 
\nm{\ket{\phi}}\abs{\dfrac{1}{\nm{\ket{\psi}}}-\dfrac{1}{\nm{\ket{\phi}}}}\nonumber\\
&=
\dfrac{\nm{\ket{\psi}-\ket{\phi}}}{\nm{\ket{\psi}}} + 
\dfrac{\abs{\nm{\ket{\psi}}-\nm{\ket{\phi}}}}{\nm{\ket{\psi}}}\nonumber\\
&\le
\dfrac{2\nm{\ket{\psi}-\ket{\phi}}}{\nm{\ket{\psi}}}\nonumber\\
&=
\dfrac{2\beta}{\alpha}
\end{align}
as claimed.
\end{proof}

\begin{lemma} Let $| \psi \rangle = \alpha | 0 \rangle | \psi_0 \rangle + \sqrt{1- \alpha^2} | 1 \rangle | \psi_1 \rangle $ and $| \phi \rangle = \beta | 0 \rangle | \phi_0 \rangle + \sqrt{1- \beta^2} | 1 \rangle | \phi_1 \rangle $, where $\ket{\psi_0}$, $\ket{\psi_1}$, $\ket{\phi_0}$, $\ket{\phi_1}$ are unit vectors, and $\alpha, \beta \in [0, 1]$. Suppose $\nm{ | \psi \rangle  - | \phi \rangle } \leq \delta < \alpha $. Then $\nm{ | \phi_0 \rangle - | \psi_0 \rangle } \leq \frac{2 \delta}{\alpha - \delta}$.
\label{lem:statedist}
\end{lemma}

\begin{proof}
First note that $\nm{ | \psi \rangle  - | \phi \rangle } \leq \delta$ implies
\begin{equation}
  \nm{ \langle 0 | \psi \rangle  - \langle 0 | \phi \rangle } \leq \delta  . \label{inequ1}
\end{equation}
In addition, by the triangle inequality, we have
\begin{equation}
\nm{ \langle 0 | \phi \rangle }  \geq  \nm{ \langle 0 | \psi \rangle } - \nm{ \langle 0 | \psi \rangle  - \langle 0 | \phi \rangle } \geq \alpha - \delta . \label{inequ2}
\end{equation}
Then Eqs.~(\ref{inequ1}) and (\ref{inequ2}) imply
\begin{equation}
\displaystyle\frac{ \nm{ \langle 0 | \psi \rangle  - \langle 0 | \phi \rangle } }{ \nm{ \langle 0 | \phi \rangle } } \leq \frac{\delta}{\alpha - \delta}. \label{inequ3}
\end{equation}
Then by the triangle inequality, we get
\begin{align}
\nm{ | \phi_0 \rangle - | \psi_0 \rangle }
& = \left\lVert \displaystyle\frac{ \langle 0 | \phi \rangle }{\nm{\langle 0 | \phi \rangle} } - \frac{ \langle 0 | \psi \rangle }{\nm{\langle 0 | \psi \rangle} } \right\rVert \nonumber\\
& = \left\lVert \displaystyle\frac{ \langle 0 | \phi \rangle}{\nm{\langle 0 | \phi \rangle} } - \dfrac{\langle 0 | \psi \rangle }{\nm{\langle 0 | \phi \rangle} } + \frac{ \langle 0 | \psi \rangle }{\nm{\langle 0 | \phi \rangle} } - \frac{ \langle 0 | \psi \rangle }{\nm{\langle 0 | \psi \rangle} } \right\rVert\nonumber \\
& = \left\lVert \displaystyle\frac{ \langle 0 | \phi \rangle -\langle 0 | \psi \rangle }{\nm{\langle 0 | \phi \rangle} } +  \langle 0 | \psi \rangle \frac{\nm{\langle 0 | \psi \rangle} - \nm{\langle 0 | \phi \rangle}  }{\nm{\langle 0 | \phi \rangle} \nm{\langle 0 | \psi \rangle}} \right\rVert \nonumber\\
& \leq \displaystyle\frac{ \nm{\langle 0 | \phi \rangle -\langle 0 | \psi \rangle} }{\nm{\langle 0 | \phi \rangle} }  +  \abs{ \frac{\nm{\langle 0 | \psi \rangle} - \nm{\langle 0 | \phi \rangle}  }{\nm{\langle 0 | \phi \rangle}} } \nonumber\\
& \leq \displaystyle\frac{ \nm{\langle 0 | \phi \rangle -\langle 0 | \psi \rangle} }{\nm{\langle 0 | \phi \rangle} }  +  \frac{\nm{\langle 0 | \psi \rangle - \langle 0 | \phi \rangle}  }{\nm{\langle 0 | \phi \rangle}}\nonumber \\
& \leq  \frac{\delta}{\alpha - \delta} +  \frac{\delta}{\alpha - \delta}\nonumber \\
&= \frac{2 \delta}{\alpha - \delta}
 \label{inequ4}
\end{align}
as claimed.
\end{proof}

\begin{lemma} Let $| \psi \rangle = \alpha | 0 \rangle | \psi_0 \rangle + \sqrt{1- \alpha^2} | 1 \rangle | \psi_1 \rangle $ and $| \phi \rangle = \beta | 0 \rangle | \phi_0 \rangle + \sqrt{1- \beta^2} | 1 \rangle | \phi_1 \rangle $, where $\ket{\psi_0}$, $\ket{\psi_1}$, $\ket{\phi_0}$, $\ket{\phi_1}$ are unit vectors, and $\alpha, \beta \in [0, 1]$. Suppose $\nm{ | \psi \rangle  - | \phi \rangle } \leq \delta < \alpha $. Then $\beta \geq \alpha - \delta$. 
\label{lem:probdist}
\end{lemma}

\begin{proof}
Note that $\nm{ | \psi \rangle  - | \phi \rangle } \leq \delta$ implies
$\nm{\braket{0}{\psi} - \braket{0}{\phi})} \le \delta$. Then by the triangle inequality, we get
\begin{align}
\beta &= \nm{\braket{0}{\phi}} \nonumber\\
&= \nm{\braket{0}{\psi} - \braket{0}{\psi} + \braket{0}{\phi}} \nonumber\\
& \ge \nm{\braket{0}{\psi}} - \nm{\braket{0}{\psi} - \braket{0}{\phi}} \nonumber\\
& \ge \alpha - \delta
\end{align}
as claimed.
\end{proof}

\end{appendix}

\end{document}